\let\oldvec\vec

\documentclass{llncs}

\let\vec\oldvec

\usepackage[ruled,vlined,linesnumbered,commentsnumbered]{algorithm2e}
\usepackage{amsmath,amssymb,amsfonts}
\begin{document}

\newtheorem{fact}[lemma]{Fact}

\title{A $(4+\epsilon)$-approximation for $k$-connected~subgraphs} % via area biset functions}
\author{Zeev Nutov}
\institute{The Open University of Israel \ \email{nutov@openu.ac.il}}

\maketitle

%%%%%%%%%%%%%%%%%%%%%%%%%%%%

\newcommand {\ignore} [1] {}

\def\kcs        {\sc $k$-Connected Subgraph}
\def\kcsa        {\sc $k$-Connectivity Augmentation}
\def\bfuc    {\sc Biset-Function Edge-Cover}

% Calligraphic 
\def \CC   {{\cal C}}
\def \FF   {{\cal F}}

% Greek letters
\def \al   {\alpha}
\def \be   {\beta}
\def \ga   {\gamma}
\def \de   {\delta}

% Bisets 
\def \A    {\mathbb{A}}
\def \B    {\mathbb{B}}
\def \C    {\mathbb{C}}
\def \p    {\partial}

% Sets
\def \empt {\emptyset}
\def \sem  {\setminus}
\def \subs {\subseteq}

% Accents etc
\def \h    {\hat}
\def \b    {\bar}

\def \f    {\frac}

\begin{abstract}
We obtain approximation ratio $2(2+\f{1}{\ell})$ for the (undirected) {\kcs} problem, 
where $\ell \approx \f{1}{2} (\log_k n-1)$ is the largest integer such that $2^{\ell-1} k^{2\ell+1} \leq n$.
For large values of $n$ this improves the $6$-appro\-xi\-ma\-tion of Cheriyan and V\'egh \cite{CVegh} when $n =\Omega(k^3)$,
which is the case $\ell=1$. For $k$ bounded by a constant we obtain ratio $4+\epsilon$.
For large values of $n$ our ratio matches the best known ratio $4$ for the augmentation version of the problem \cite{N-C,N-A},
as well as the best known ratios for $k=6,7$ \cite{KN1}.
Similar results are shown for the problem of covering an arbitrary crossing supermodular biset function.
\end{abstract}

\section{Introduction}

% Two paths in a graph are {\bf internally disjoint} if no internal node of one of the paths belongs to the other.
A graph is {\bf $k$-connected} if it contains $k$ internally disjoint paths from every node to any other node.
We consider the following problem on undirected graphs:

\begin{center} \fbox{\begin{minipage}{0.965\textwidth} \noindent
\underline{\kcs} \\
{\em Input:} \ \ A graph $\h G=(V,\h E)$ with edge costs $\{c_e:e \in \h E\}$ and an integer $k$. \\
{\em Output:}   A minimum cost $k$-connected spanning subgraph of $\h G$. 
\end{minipage}}\end{center}

For undirected graphs, the problem is NP-hard for $k=2$ (the case $k=1$ is the {\sc MST} problem),
even when all edges in $\h G$ have unit costs.
This is since any feasible solution with $n=|V|$ edges is a Hamiltonian cycle.
For directed graphs the problem is NP-hard already for $k=1$, by a similar reduction. 
Let $\vec{\rho}(k,n)$ and $\rho(k,n)$ denote the best possible approximation ratios for the directed and undirected variants, respectively.
A standard bi-direction reduction gives $\rho(k,n) \leq 2\vec{\rho}(k,n)$,
while by \cite{LaN} $\rho(k,n) \geq \vec{\rho}(k-n/2,n/2)$ for $k \geq n/2+1$.
All in all we get that for $k \geq n/2+1$ the approximability of directed and undirected cases is the same, up to a $2$ factor.
This however does not exclude that the undirected case is easier when $n $ is much larger than $k$.

In the {\kcsa} problem $\h G$ contains a spanning $(k-1)$-connected subgraph of cost $0$.
A feasible solution for the {\kcs} problem can be obtained by solving sequentially {\sc $\ell$-Connectivity Augmentation} problems
for $\ell=1, \ldots k$, but this reduction usually invokes a factor of $H(k)$ in the ratio, where $H(k)$ denotes the $k$-th harmonic number.
Several ratios for {\kcs} were obtained in this way, c.f. \cite{CVV,KN2,FL,N-C}.
The currently best known ratios for the general and the augmentation version, 
for both directed and undirected graphs, are summarized in Table~1. 

\begin{table}[htbp] \small 
\begin{center} 
\begin{tabular}{|l|l||l|l|l|l|} \hline
\multicolumn{2}{|c||}{\em General}                                                                                                                              & \multicolumn{2}{|c|}{\em Augmentation}
\\\hline        
{\em Undirected}                                                              & {\em Directed}                                                                    & {\em Undirected}           & {\em Directed}
\\\hline
${O\left(\ln \frac{n}{n-k} \cdot \ln k\right)}$ \cite{N-C} & ${O\left(\ln \frac{n}{n-k} \cdot \ln k\right)}$ \cite{N-C}   & $2H(\mu)+1$ \cite{N-A} & $H(\mu)+3/2$ \cite{N-A} \\
$6$ if $n \geq k^3$ \cite{CVegh} (see also \cite{FNR})   &                                                                                              & $O(\ln(n-k))$ \cite{N-C} & $O(\ln(n-k))$ \cite{N-C}   \\
$\lceil(k+1)/2\rceil$ if $k \leq 7$ \cite{ADNP,DN,KN1}   & $k+1$ if $k \leq 6$ \cite{KN1}                                             &                                         & 
\\\hline
\end{tabular}
\newline 
\caption{\small
Known approximation ratios for {\kcs} and {\kcsa} problems.
Here $\mu= \left\lfloor\f{n}{\lfloor (n-k)/2 \rfloor +1}\right\rfloor$, and note that if $\mu=1$ then $k \in \{1,2\}$,
and that if $n \geq 3k-2$ then $\mu=2$ and $H(\mu)=3/2$.}
\end{center}
\vspace*{-0.5cm}
\end{table}

Note that we consider the {\em node-connectivity} version of the problem, for which classic techniques 
like the primal dual method \cite{GGPS} and the iterative rounding method \cite{Jain,FJW} 
do not seem to be applicable directly.
Ravi and Williamson \cite{RWe} gave an example of a {\kcsa} instance 
when the primal dual method has ratio $\Omega(k)$.
A related example of Aazami, Cheriyan and Laekhanukit \cite{ACL} rules out the iterative rounding method.
On the other hand, several works showed that the problem can be decomposed into
a small number of ``good''  sub-problems. 
However, attempts to achieve a constant ratio for $k = n-o(n)$ (e.g., for $k=n-\Theta(\sqrt{n})$) failed even in the easier augmentation case,
thus Cheriyan and V\'egh \cite{CVegh} suggested the following question: \\
{\em What ratio can we achieve when $n$ is lower bounded by a function of $k$?} \\
This essentially addresses the issue of ''asymptotic approximability'' -- as a function of the single parameter $k$, for all sufficiently large $n$. 
For undirected graphs Cheriyan and V\'egh \cite{CVegh} % called this the ``asymptotic setting''.
gave an elegant $6$ approximation when $n \geq k^4$, and this bound was slightly improved to $n \geq k^3$ in \cite{FNR}.

Note that the ``asymptotic approximability'' question seems almost settled for the augmentation version:
by \cite{N-C,N-A} for both directed and undirected graphs we have a constant ratio unless $k=n-o(n)$;
furthermore, for undirected graphs we have ratio $4$ for $n \geq 3k-2$ (ratio $3$ for directed graphs) \cite{N-A},
and this is also the best known ratios when $k=6,7$ for the general version \cite{KN1}.

From now and on we consider undirected graphs, unless stated otherwise.
Note that $4$ is a current ``lower bound'' on the ``asymptotic approximability'' of the problem, 
in the sense that no better ratio is known for much easier sub-problems.
Our main result shows that this ``lower bound'' is (almost) achievable.

\begin{theorem} \label{t:kcs}
{\kcs} admits approximation ratio $2(2+1/\ell)$ where $\ell$ is the largest integer such that $n \geq k[(k^2-1)(2k^2-3k+2)^{\ell-1}+1]$.
\end{theorem}

Note that $\ell \approx \f{1}{2} (\log_k n-1)$ and that Theorem~\ref{t:kcs} implies 
ratio $5$ for $n \geq 2k^5$ and ratio $4+\epsilon$ if $k$ is bounded by a constant. 
In fact, we prove a generalization of Theorem~\ref{t:kcs}, stated in biset function terms, see the next section.

We note that our result can be used to improve approximation ratios for the {\sc Min-Cost Degree Bounded} {\kcs} problem, see \cite{FNR,EV}.

We refer the reader to \cite{N-sn,N-kcs} for surveys on approximation algorithms for node-connectivity problems,
and to  \cite{Frank-book,FK-survey} for a survey on polynomially solvable cases. 
Here we briefly mention the status of some restricted {\kcs} problems.

\begin{table}[htbp]   \small 
\begin{center} 
\begin{tabular}{|l|l|l|}                    
\hline        
{\em Costs}    & \hspace{1.2cm} {\em Undirected}                                                                                      & \hspace{1.2cm} {\em Directed}                                      
\\\hline
\hline
$\{0,1\}$         & $\min\{2,1+\frac{k^2}{2{\sf opt}}\}$ \cite{FJ,JJ1}                                                               & in P \cite{FJ}    
\\\hline
$\{1,\infty\}$ & $1-\frac{1}{k}+\frac{n}{\sf opt} \leq 1+\frac{1}{k}$ \cite{CT00} (see also \cite{N-small})  & $1-\frac{1}{k}+\frac{2n}{\sf opt} \leq 1+\frac{1}{k}$ \cite{CT00} (see also \cite{N-small}) 
\\\hline
metric            & $2+(k-1)/n$ \cite{KN1}                                                                                                          & $2+k/n$ \cite{KN1}                
\\\hline
\end{tabular}
\newline 
\caption{\small
Known approximation ratios of {\kcs} problems.}
\end{center} \vspace*{-0.5cm}
\end{table}

We may assume that the input graph $\h G$ is complete, by assigning infinite costs to ``forbidden'' edges. 
Under this assumption, except for general edge costs, three main types of costs are considered in the literature:
\begin{itemize}
\item
{$\{0,1\}$-costs:}  
Here we are given a graph $G$, and the goal is to find 
a minimum size set $J$ of new edges (any edge is allowed) such that $G \cup J$ is $k$-connected. 
\item
{$\{1,\infty\}$-costs:} 
Here we seek a $k$-connected spanning subgraph of $\h G$ with minimum number of edges. 
\item
{metric costs:} 
The costs satisfy the triangle inequality $c_{uv} \leq c_{uw}+c_{wv}$ for all $u,w,v \in V$.
\end{itemize}

The currently best known appro\-xi\-mation ratios for these costs types are summarized in Table~2,
and we mention some additional results.
For $\{0,1\}$-costs the complexity status of the problem is not known for undirected graphs, 
but for any constant $k$ an optimal solution can be computed in polynomial time \cite{JJ2}. 
When $\h G$ contains a spanning $(k-1)$-connected subgraph of cost $0$
the $\{0,1\}$-costs case can be solved in polynomial time for any $k$ \cite{Vegh}.
In the case of $\{1,\infty\}$-costs, directed $1$-{\sc Connected Subgraph} admits ratio $3/2$ \cite{Vetta-scs}.
In the case of metric costs $2$-{\sc Connected Subgraph} admit ratio $3/2$ \cite{FJ82}.

%%%%%%%%%%%%%%%%%%%%%%%%%%%%%%%%%%%
\section{Biset functions and $k$-connectivity problems} \label{s:bisets}
%%%%%%%%%%%%%%%%%%%%%%%%%%%%%%%%%%%

While edge-cuts of a graph correspond to node subsets, a natural way to
represent a node-cut of a graph is by a pair of sets called a ``biset''.

\begin{definition}
An ordered pair $\A=(A,A^+)$ of subsets of $V$ with $A \subs A^+$ is called a {\bf biset}; 
$A$ is the {\bf inner part} and $A^+$ is the {\bf outer part} of $\A$,
and $\p\A=A^+ \sem A$ is the {\bf boundary} of $\A$. 
The {\bf co-set} of $\A$ is $A^*=V \sem A^+$; the {\bf co-biset} of $\A$ is $\A^*=(A^*,V \sem A)$. 
We say that $\A$ is {\bf void} if $A=\empt$, {\bf co-void} if $A^+=V$ (namely, if $A^*=\empt$), and $\A$ is {\bf proper} otherwise.
% Let $\VV$ denote the family of bisets over $V$.
\end{definition}

A {\bf biset function} assigns to every biset $\A$ a real number; in our context, it will always be an integer (possibly negative).

\begin{definition}
An {\bf edge covers a biset} $\A$ if it goes from $A$ to $A^*$.
For an edge-set/graph $J$ let $\delta_J(\A)$ denote the set of edges in $J$ covering $\A$. 
The {\bf residual function} of a biset function $f$ w.r.t. $J$ is defined by $f^J(\A)=f(\A)-|\de_J(\A)|$. 
We say that $J$ {\bf $f$-covers} $\A$ if $|\de_J(\A)| \geq f(\A)$, and we say that
{\bf $J$ covers $f$} or that $J$ is an {\bf $f$-cover} if $|\de_J(\A)| \geq f(\A)$ for all $\A$.
\end{definition}

In biset terms, Menger's Theorem says that the maximum number of internally disjoint $st$-paths in $G$
equals to $\min\{|\p\A|+|\de_G(\A)|:s \in A,t \in A^*\}$.
Consequently, $G$ is $k$-connected if and only if $|\de_G(\A)| \geq k-|\p\A|$ for every proper biset $\A$;
note that non-proper bisets cannot and are not required to be covered.
Thus $G$ is $k$-connected if and only if $G$ covers the {\bf $k$-connectivity biset function} $f_{k{\sf\mbox{-}CS}}$
defined by 
\begin{equation*} % \label{e:fkcs}
f_{k{\sf\mbox{-}CS}}(\A)= 
\left \{ \begin{array}{ll}
               k-|\p\A| \ \ \     & \mbox{if } \A \mbox{ is proper} \\
               0                  & \mbox{otherwise}
\end{array} \right .
\end{equation*}
We thus will consider the following generic problem:

\begin{center} \fbox{\begin{minipage}{0.960\textwidth} \noindent
\underline{\bfuc} \\
{\em Input:} \ A graph $\h G=(V,\h E)$ with edge costs and a biset function $f$ on $V$. \\
{\em Output:} A minimum cost edge-set $E \subs \h E$ that covers $f$.
\end{minipage}}\end{center}

Here $f$ may not be given explicitly, and an efficient implementation  
of algorithms requires that certain queries related to $f$ can be answered in time polynomial in $n$.
We will consider later implementation details.
In the application discussed here, relevant polynomial time oracles are available via min-cut computations. 
In particular, we have a polynomial time separation oracle for the LP-relaxation
due to Frank and Jord\'{a}n~\cite{FJ}: 
\[\begin{array} {llllllll} 
&                                                              & \tau(f)= & \min           & c \cdot x                                                              &                                             &\\
& \mbox{\bf(Biset-LP)} \hspace{2.0cm} &               &\mbox{s.t.} & x(\delta_{\h E}(\A)) \geq f(\A) \hspace{0.3cm} & \forall \A  \hspace{5.0cm} &\\
&                                                              &               &                    & 0 \leq x_e \leq 1                                                 & \forall e                              & 
\end{array}\]

This LP is particularly useful if the biset function $f$ has good uncrossing/supermodularity properties.
To state these properties, we need to define the intersection and the union of bisets.

\begin{definition} \label{d:bisets}
The {\bf intersection} and the {\bf union} of two bisets $\A,\B$ are defined by
$\A \cap \B = (A \cap B, A^+ \cap B^+)$ and $\A \cup \B=(A \cup B,A^+ \cup B^+)$.
The biset $\A \sem \B$ is defined by $\A \sem \B=(A \sem B^+,A^+ \sem B)$.
% We say that {\bf $\B$ contains $\A$} and write $\A \subs \B$ if $A \subs B$ and $A^+ \subs B^+$. 
We say that $\A,\B$ {\bf intersect} if $A \cap B \neq \empt$, and 
{\bf cross} if $A \cap B \neq \empt$ and $A^+ \cup B^+ \neq V$.
\end{definition}

The following properties of bisets are easy to verify. 

\begin{fact} \label{f:AB}
For any bisets $\A,\B$ the following holds.
If a directed/undirected edge $e$ covers one of $\A \cap \B,\A \cup \B$ then $e$ covers one of $\A,\B$; 
if $e$ is an undirected edge, then if $e$ covers one of $\A \sem \B,\B \sem \A$, then $e$ covers one of $\A,\B$.
Furthermore                    
$|\p\A|+|\p\B| = |\p(\A \cap \B)|+|\p(\A \cup \B)|=|\p(\A \sem B)|+|\p(\B \sem \A)|$.
\end{fact}

For a biset function $f$ and bisets $\A,\B$ the {\bf supermodular inequality} is
$$
f(\A \cap \B) +f(\A \cup \B) \geq f(\A)+f(\B) \ .
$$
We say that a biset function $f$ is {\bf supermodular} if the supermodular inequality holds for all $\A,\B$,
and {\bf modular} if the supermodular inequality holds as equality for all $\A,\B$;
$f$ is {\bf symmetric} if $f(\A)=f(\A^*)$ for all $\A$.
Using among others Fact~\ref{f:AB}, one can see the following.
\begin{itemize}
\item
For any directed/undirected graph $G$ the function $-d_G(\cdot)$ is supermodular. 
\item 
The function $|\p(\cdot)|$ is modular.
\item
For any $R \subs V$ the function $|A \cap R|$ is modular.  
\end{itemize}

We say that a biset $\A$ is {\bf $f$-positive} if $f(\A)>0$.
Some important types of biset functions are given in the following definition.

\begin{definition} \label{d:types}
A biset function $f$ is {\bf intersecting/crossing supermodular} if the supermodular inequality holds 
whenever $\A,\B$ intersect/cross;
$f$ is {\bf positively intersecting supermodular} if the supermodular inequality holds for any pair of intersecting $f$-positive bisets.
\end{definition}

{\bfuc} with positively intersecting supermodular $f$ admits a polynomial time algorithm due to Frank \cite{F-R}
that for directed graphs computes an $f$-cover of cost $\tau(f)$  (this also can be deduced using the iterative rounding method); 
for undirected graphs the cost is at most $2\tau(f)$, by a standard bi-direction reduction.
Note however that the function $f_{k{\sf\mbox{-}CS}}$ that we want to cover is obtained by zeroing an intersecting supermodular
function on co-void bisets, but $f_{k{\sf\mbox{-}CS}}$ itself is not positively intersecting supermodular.

In general, changing a supermodular function on void bisets gives an intersecting supermodular function, while
changing an intersecting supermodular function on co-void bisets gives a crossing supermodular function
(not all crossing supermodular functions arise in this way -- see \cite{FK-survey}).
In particular, zeroing a supermodular function on non-proper bisets gives a crossing supermodular function.
For example, the $k$-connectivity function $f_{k{\sf\mbox{-}CS}}$ is obtained in this way  
from the modular function $k-|\p\A|$, thus $f_{k{\sf\mbox{-}CS}}$ is crossing supermodular.

A common way to find a ``cheap'' partial cover of $f_{k{\sf\mbox{-}CS}}$ is to find a $2$-approximate cover of the 
{\bf fan function} $g_R$ obtained by zeroing the function $k-|\p\A|-|A \cap R|$ on void bisets, where $R \subs V$ with $|R|=k$.
Note that $g_R$ is intersecting supermodular and that $g_R$ is non-positive on co-void bisets (e.g., $g_R((V,\empt))=k-0-k=0$).
Fan functions were used in many previous works on {\kcs} problems 
starting from Khuller and Raghavachari \cite{KR}, and also by Cheriyan and V\'{e}gh \cite{CVegh}. 
In fact, covering $g_R$ is equivalent to the following connectivity problem.
Let us say that a graph is {\bf $k$-in-connected to $r$} if it has $k$ internally disjoint $vr$ paths for every $v \in V$.
Construct a graph $G_r$ by adding to $\h G$ a new node $r$ and a set $F_r$ of zero cost edge from each $v \in R$ to $r$; 
then $H=(V \cup \{r\},J_r)$ is a $k$-in-connected to $r$ spanning subgraph of $G_r$ if and only if $J=J_r \sem F_r$ covers $g_R$.
The problem of finding an optimal $k$-in-connected spanning subgraph can be solved in strongly polynomial time 
for directed graphs \cite{FT} (see also \cite{F-R}), and this implies a $2$-approximation for undirected graphs.

Fan functions are considered as the ``strongest'' intersecting supermodular functions for the purpose of finding a partial cover of $f_{k{\sf\mbox{-}CS}}$.
However, an inclusion minimal directed cover $J$ of a fan function may be difficult to decompose,
since $J$ may have directed edges with tail in $R$;
this is so since a fan function requires to cover to some extent bisets $\A$ with $A \cap R\neq \empt$.
We therefore use a different type of functions defined below, that are ``weaker'' but have ``better'' decomposition properties.

For $R \subs V$ the {\bf area function} of $f$ is defined by 
$$
f_R(\A)=f(\A)-\max_\A f(\A) \cdot |A \cap R| \ .
$$
Note that $f_R(\A)=f(\A)$ if $A \subs V \sem R$ and $f_R(\A) \leq 0$ otherwise,
so $f_R$ requires to $f$-cover only those bisets whose inner part is contained in the ``area'' $V \sem R$.
In the next two lemmas we give some properties of area functions. Let us denote
$$k_f = 1+\max\{|\p\A|:f(\A)>0\} \ .$$ 

\begin{lemma}
If $|R| \geq k_f$ then: $f_R$ is non-positive on co-void bisets,
$f_R$ is intersecting supermodular if $f$ is, and 
$f_R$ is positively intersecting supermodular if $f$ is crossing supermodular and $|R| \geq 2k_f-1$.
\end{lemma}
\begin{proof}
The first two statements are easy, so we prove only the last statement. Let $\A,\B$ be intersecting $f_R$-positive bisets.
Then $A \cap R=B \cap R=\empt$, and thus $(A \cap B) \cap R=(A \cup B) \cap R=\empt$.
Consequently, $f_R=f$ on the bisets $\A,\B,\A \cap \B,\A \cup \B$. 
Moreover, $A^* \cap B^* \cap R \neq \empt$, since $|\p\A \cup \p\B| < 2(k_f-1)+1 \leq |R|$. 
Thus $\A,\B$ must cross, and since $f$ is crossing supermodular
$$
f_R(\A)+f_R(\B) = f(\A)+f(\B) \leq f(\A \cap \B)+f(\A \cup \B) = f_R(\A \cap \B)+f_R(\A \cup \B) \ .
$$
Consequently, the supermodular inequality holds for $\A,\B$ and $f_R$.
\qed
\end{proof}

For $S \subs V$ let $\ga(S)$ denote the set of edges in $\h E$ with both end in $S$. 
Consider the following algorithm for covering $f_R$. % where $|R| \geq k_f$:  

\medskip \medskip

\begin{algorithm}[H]
\caption{{\sc Area-Cover$(\h{G},c,f,R)$}} 
\label{alg:AC} 
bidirect the edges in $\ga(V \sem R)$ and direct into $R$ the edges in $\de(R)$ \\ 
compute an optimal directed edge-cover $I'$ of $f_R$  \\  
return the underlying undirected edge set $I$ of $I'$
\end{algorithm}

\medskip \medskip

If $f_R$ is positively intersecting supermodular, then step~2 in the algorithm
can be implemented in polynomial time if the Biset-LP for $f_R$ can be solved in polynomial time.
In many specific cases strongly polynomial algorithms are available. 
E.g., if $f$ is obtained by zeroing the function $k-|\p\A|$ on void bisets 
then we can use the Frank-Tardos algorithm \cite{FT} 
or the algorithm of Frank \cite{F-R} for finding a directed min-cost $k$-in-connected subgraph -- 
in the above reduction described for covering a fan function $g_R$, 
the edge set $F_r$ will have $k=\max_\A f(\A)$ parallel directed edges from each $v \in R$ to the root $r$.

The following lemma relates the cost of the solution computed by Algorithm~\ref{alg:AC} to the Biset-LP value.

\begin{lemma} \label{l:area}
Let $f_R$ be positively intersecting supermodular and let $x$ be a feasible Biset-LP solution for covering $f_R$.
Then Algorithm~\ref{alg:AC} returns an $f_R$-cover $I$ of cost
$c(I) \leq \sum_{e \in \de(R)}c_ex_e+2\sum_{e \in \ga(V \sem R)}c_ex_e$.
\end{lemma}
\begin{proof}
Edges in $\ga(R)$ do not cover $f_R$-positive bisets, hence they can be removed.
Let $E'$ be the bi-direction of $\h E$, where each undirected edge $e$ with ends $u,v$ is 
replaced by two arcs $uv,vu$ of cost $c_e$ and value $x_e$ each. 
Let $x'$ be be obtained by zeroing the value of arcs leaving $R$; these arcs do not cover $f$-positive bisets.
We claim that:
$$
c(I) \leq c(I') \leq \sum_{e \in E'} c_e x'_e  =\sum_{e \in \de(R)}c_ex_e+2\sum_{e \in \ga(V \sem R)}c_ex_e
$$
The first inequality is obvious. The second inequality is since $f_R$ is positively intersecting supermodular and 
since $x'$ is a directed feasible Biset-LP solution for $f_R$ while $I'$ is an optimal one.
The equality is by the construction.
\qed
\end{proof}

Assuming that for any residual function of $f^I$ of $f$, Algorithm~\ref{alg:AC} can be implemented in polynomial time 
whenever $f^I_R$ is positively intersecting supermodular, 
and that the Biset-LP for covering $f^I$ can be solved in polynomial time, 
we prove the following theorem that implies Theorem~\ref{t:kcs}.

\begin{theorem} \label{t:bfuc}
Undirected {\bfuc} with symmetric cros\-sing supermodular $f$ admits approximation ratio $2(2+1/\ell)$,
where $\ell$ is the largest integer such that:
\begin{itemize}
\item
$n \geq (2k_f-1)[(2k_f^2-3k_f+2)^\ell +1]$ for symmetric crossing supermodular $f$.
\item
$n \geq k_g[(2k_g^2-3k_g+2)^\ell+1]$ if $f$ is obtained by zeroing an intersecting supermodular function $g$ on co-void bisets, 
\item
$n \geq k[(k^2-1)(2k^2-3k+2)^{\ell-1}+1]$ if $f=f_{k{\sf\mbox{-}CS}}$. 
\end{itemize}
\end{theorem}

%%%%%%%%%%%%%%%%%%%%%%%%%%%%%%
\section{Covering crossing supermodular functions (Theorem~\ref{t:bfuc})} \label{s:bfuc}
%%%%%%%%%%%%%%%%%%%%%%%%%%%%%% 

A biset function $f$ is {\bf positively skew-supermodular} if the supermodular inequality or 
the {\bf co-supermodular inequality} $f(\A \sem \B)+f(\B \sem \A) \geq f(\A)+f(\B)$ holds for $f$-positive bisets.

The corresponding {\bfuc} problem, when $f$ is posi\-tively skew-supermodular, admits ratio $2$ 
(assuming the Biset-LP can be solved in polynomial time) 
\cite{FJW}; see also \cite{CVegh,FNR} for a simpler proof along the proof line of \cite{NRS} for the set functions case.

We say that $\A,\B$ {\bf co-cross} if $\A \sem \B$ and $\B \sem \A$ are both non-void,
and that $\A,\B$ {\bf independent} if they do not cross nor co-cross.
One can verify that $\A,\B$ are independent if and only if at least one of the following holds:
$A \subs \p\B$, or $A^* \subs \p\B$, or $B \subs \p\A$, or $B^* \subs \p\A$.
A biset function $f$ is {\bf independence-free} if no pair of $f$-positive bisets are independent. 
It is easy to see that if $f$ is symmetric and if $|A| \geq k_f$ holds for every $f$-positive biset $\A$ then $f$ is independence-free.

\begin{lemma} [\cite{JJ2}] \label{l:co}
Let $f$ be a symmetric crossing supermodular biset function. If $\A,\B$ are not independent
then the supermodular or the co-supermodular ine\-quality holds for $\A,\B$ and $f$.
Thus if $f$ is independence-free then $f$ is positively skew-supermodular. 
\end{lemma}
\begin{proof}
If $\A,\B$ cross then the supermodular inequality holds for $\A,\B$. Assume that $\A,\B$ co-cross.
Then $\A$ and $\B^*$ cross, and thus the supermodular inequality holds for $\A,\B^*$ and $f$. 
Note that (i) $\A \sem \B=\A \cap \B^*$; (ii) $\A \cup \B^*$ is the co-biset of $\B \sem \A$,
hence $f(\A \cup \B^*)=f(\B \sem \A)$, by the symmetry of $f$. Thus we get
$f(\A \sem \B)+f(\B \sem \A)=f(\A \cap \B^*)+f(\A \cup \B^*) \geq f(\A)+f(\B^*)=f(\A)+f(\B)$.
\qed
\end{proof}

This suggests a two phase strategy for covering an ``almost'' supermodular function $f$.
First, find a ``cheap'' edge set $J$ such that the residual function $f^J$ will be independence-free so $f^J$ will have ``good uncrossing properties''. 
Second, use some ``known'' algorithms to cover $f^J$. 
The idea is due to Jackson and Jord\'{a}n \cite{JJ2}, and it is also the basis of the algorithm of Cheriyan and V\'{e}gh \cite{CVegh}
(see also \cite{N-R} where the same idea was used for a related problem).
Specifically, if $f$ is crossing supermodular,
we will seek a cheap $J$ that $f$-covers all bisets $\A$ with $|A| \leq k_f-1$;
by Lemma~\ref{l:co} the residual function $f^J$ will be positively skew-supermodular
so we can use the $2$-approximation algorithms of \cite{FJW} to cover~$f^J$. 

% Thus to prove Theorem~\ref{t:bfuc} it is sufficient to prove the following:

% \begin{lemma} \label{l:ell}
% Undirected {\bfuc} with crossing super\-mo\-du\-lar $f$ admits a polynomial time algorithm that computes an edge set $J$
% with $c(J) \leq 2(1+1/\ell) \cdot \tau(f)$ such that $|A| \geq k_f$ holds for any $f^J$-positive biset $\A$.
% \end{lemma}

The algorithm of Cheriyan and V\'{e}gh \cite{CVegh} finds an edge $J$ as above of cost $\leq 4\tau(f)$, by covering two fan functions. 
Our algorithm covers a pair of area functions. In fact, we will cover a sequence of $\ell \geq \ell$ pairs of area functions, 
and with the help of Lemma~\ref{l:area} will show that the sum of their costs is at most $2\tau(\ell'+1)$;
we choose the cheapest pair cover that will have cost $\leq 2\tau(1+1/\ell')$.

For an integer $p$ let $U(f,p)=\bigcup\{A:f(\A)>0,|A| \leq p\}$ be the union of inner parts of size $\leq p$ of the $f$-positive bisets.
Note that if $R \subs V \sem U(f,p)$ and if $I$ covers $f_R$ then $f^I(\A) \leq 0$ whenever $|A| \leq p$.
Thus from Lemma~\ref{l:co} we get:

\begin{corollary} \label{c:cross-skew}
If $R \subs V \sem U(f,k_f)$ and if $I'$ is an $f_R$-cover, then the resi\-dual function $f^{I'}$ of $f$ w.r.t. $I'$ 
is independence-free and thus is positively skew-supermodular.
\end{corollary}

Thus we just need to find $R \subs V \sem U(f,k_f)$ with $|R| \geq k_f$ and 
compute a $2$-approximate cover of $f_R$ -- the residual function will be independence-free and thus positively skew-supermodular.
However, such $R$ may not exist, e.g., for $f=f_{k{\sf\mbox{-}CS}}$ we have $U(f,k)=V$.
The idea of Cheriyan and V\'{e}gh \cite{CVegh} resolves this difficulty as follows: 
first find a ``cheap''  edge set $I$ such that $|U(f^I,k_f)| \leq n-k_f$ will hold for the residual function $f^I$,
and only then compute for $f^I$ an edge set $I'$ as in Corollary~\ref{c:cross-skew}.
Then the function $f^{I \cup I'}$ is independence-free and thus is positively skew-supermodular.

Variants of the next lemma were proved in \cite{CVegh,FNR} (our bound is just slightly better), and we use it to show that $I$ as above 
can be a cover of an area function, provided that $n$ is large enough.
Let us say that a biset family $\FF$ is {\bf weakly posi-uncrossable} if for any $\A,\B \in \FF$ such that 
both bisets $\A \sem \B,\B \sem \A$ are non-void, one of them is in $\FF$. If $f$ is crossing
supermodular and symmetric then the family $\FF$ of $f$-positive bisets is weakly posi-uncrossable, see \cite{CVegh,FNR}.

\begin{lemma} [\cite{CVegh,FNR}] \label{l:FNR}
Let $\FF$ be a weakly posi-uncrossable biset family, 
let $p=\max_{\A \in \FF} |A|$, $q=\max_{\A \in \FF} |\p\A|$, $U=\bigcup_{\A \in \FF} A$,
and let $\nu$ be the maximum number of pairwise inner part disjoint bisets in $\FF$.
Then $|U| \leq \nu[(2q(p-1)+p]$.
\end{lemma}
\begin{proof}
Let $\FF'$ be an inclusion minimal subfamily of $\FF$ such that $\bigcup_{\A \in \FF'} A = U$.
By the minimality of $|\FF'|$, for every $\A_i \in \FF'$ there is $v_i \in A_i$ such that $v_i \notin A_j$ for every $j \neq i$. 
For every $i$ let $\C_i$ be an inclusion minimal member of the family $\{\C \in \FF: \C \subs \A_i,v_i \in C\}$,
where here $\A \subs \B$ means that $A \subs B$ and $A^+ \subs B^+$.
Since $\FF$ is weakly posi-uncrossable, the minimality of $\C_i$ implies that exactly one of the following holds for any $i \neq j$: 
\begin{itemize}
\item
$v_i \in \p\C_j$ or $v_j \in \p\C_i$;
\item
$\C_i=\C_i \sem \C_j$ or $\C_j=\C_j \sem \C_i$. 
\end{itemize}
Construct an auxiliary directed graph ${\cal J}$ on node set $\CC=\{\C_i:\A_i \in \FF'\}$.
Add an arc $\C_i\C_j$ if $v_i \in \p\C_j$. The in-degree in ${\cal J}$ of a node $\C_i$ is at most
$|\p\C_i| \leq q$.
Thus every subgraph of the underlying graph of ${\cal J}$ has a node of degree $\leq 2q$.
A graph is $d$-degenerate if every subgraph of it
has a node of degree $\leq d$. It is known that any $d$-degenerate graph is $(d+1)$-colorable.
Hence ${\cal J}$ is $(2q+1)$-colorable, so its node set can be partitioned into at most $2q+1$ independent sets,
say $\CC_1, \CC_2, \ldots$, where the bisets in each independent set are pairwise inner part disjoint.
W.l.o.g. we may assume that $\CC_1$ is a maximal subfamily in $\CC$ of pairwise inner part disjoint bisets,
so any $\C \in \CC \sem \CC_1$ intersects some biset in $\CC_1$.
Let $\FF'_i$ be the subfamily of $\FF'$ that corresponds to $\CC_i$, so $|\FF'_i|=|\CC_i| \leq \nu$.
Let $U_i=\cup_{\A \in \FF_i} A$.
An easy argument shows that $|U_1| \leq \nu p$ and that $|U_i \sem U_1| \leq \nu(p-1)$ for $i \geq 2$.
Consequently, $|U| \leq \nu p +2q \nu(p-1)$, as claimed.
\qed
\end{proof}

\begin{corollary} \label{c:size}
If $f$ is symmetric crossing supermodular and if $I$ is a cover of $f_R$ then $|U(f^I,k_f) \cup R| \leq |R|(2k_f^2-3k_f+2)$.
\end{corollary}
\begin{proof}
Denote $r=|R|$, $r'=|U(f^I,k_f) \cap R|$, and $k=k_f$.
Substituting $q+1=p=k_f$ and observing that $\nu \leq r'$ in Lemma~\ref{l:FNR} we get
$$
|U(f^I,k_f) \cup R| \leq r'[(2(k-1)^2+k]+(r-r') \leq r[(2(k-1)^2+k]=r(2k^2-3k+2)
$$
as required.
\qed
\end{proof}

Let us skip for a moment implementation details, 
and focus on bounding the cost of an edge set $J$ computed  by the following algorithm.

\medskip \medskip

\begin{algorithm}[H]
\caption{{\sc Growing-Cover$(\h{G},c,f)$}}
\label{alg:GC} 
let $\empt \neq R_1 \subset V$ \\
\For{$i=1$ to $\ell$}
{
$I \gets${\sc Area-Cover$(\h{G},c,f,R_i)$} \\
$R_{i+1} \gets U(f^I,k_f) \cup R_i$ \\
$I' \gets${\sc Area-Cover$(\h{G},c,f,V \sem R_{i+1})$} \\ 
% compute a $2$-approximate cover $F_i$ of $f^{I \cup I'}$ using the algorithm of \cite{FJW} \\
$J_i \gets I \cup I'$ 
}
return the cheapest edge set $J$ among the edge sets $J_1, \ldots,J_\ell$ computed
\end{algorithm}

\medskip \medskip

Let us fix some optimal Biset-LP solution $x$. 
For an edge set $F$ the $x$-cost of $F$ is defined as $\sum_{e \in F} c_ex_e$.  Let us use the following notation:
\begin{itemize}
\item
$\tau=\sum_{e \in \h E} c_ex_e$ is the optimal solution value.
\item
$\ga_i$ is the $x$-cost of the edges with both ends in $R_i$.
\item
$\de_i$ is the $x$-cost of the edges with one end in $R_i$ and the other in $V \sem R_i$.
\item
$\b{\ga}_i$ is the $x$-cost of the edges with both ends in $V \sem R_i$.
\end{itemize}

Clearly, for any $i$ we have
$$
\tau=\ga_i+\de_i+\b{\ga}_i
$$
By Lemma~\ref{l:area}, the cost of the covers $I,I'$ computed at iteration $i$ is bounded by
\begin{eqnarray*}
c(I) & \leq & \de_i+2\b\ga_i \\
c(I') & \leq & \de_{i+1}+2\ga_{i+1}
\end{eqnarray*}
Thus we get 
\begin{eqnarray*}
c(J_i) & \leq & (\de_i+2\b{\ga}_i)+(\de_{i+1}+2\ga_{i+1}) \\
         &   =   & (\de_i+\b\ga_i+\ga_i)+(\b\ga_i-\ga_i) + (\de_{i+1}+\ga_{i+1}+\b\ga_{i+1})-(\b\ga_{i+1}-\ga_{i+1})  \\
         &   =   & 2\tau+(\b{\ga}_i-\ga_i)-(\b{\ga}_{i+1}-\ga_{i+1}) 
\end{eqnarray*}
Summing this over $\ell$ iterations and observing that the sum is telescopic we get 
\begin{eqnarray*}
\sum_{i=1}^\ell c(J_i) & \leq & 2\ell \tau+\sum_{i=1}^\ell [(\b{\ga}_i-\ga_i)-(\b{\ga}_{i+1}-\ga_{i+1})] \\
                                    &   =   & 2\ell \tau +(\b{\ga}_1-\ga_1)-(\b{\ga}_{\ell+1}-\ga_{\ell+1})  \\
																		&   =   & 2\ell \tau +(\b{\ga}_1+\ga_{\ell+1})-(\ga_1+\b{\ga}_{\ell+1}) \\
																		&   =   & 2\tau(\ell+1)-(2\ga_1+\de_1+\b{\ga}_{\ell+1}+\de_{\ell+1})
\end{eqnarray*}
Thus there exists an index $i$ such that
$$
c(J_i) \leq 2\tau(1+1/\ell)
$$
Note that if $R_{i+1}=R_i$ for some $i$ then $c(J_i) \leq c(I)+c(I') \leq 2\de_i+2\ga_i+2\b\ga_i = 2\tau$, 
hence in this case the algorithm can terminate with $J=J_i$ and $c(J) \leq 2\tau$.

Next we use Corollary~\ref{c:size} to lower bound $n$ to ensure that the algorithm will have $\ell$ iterations.
Let $r=|R_1|$, and $r$ is also a lower bound on $n-|R_\ell|$.
To see the bounds on $n$ in Theorem~\ref{t:bfuc} note the following.
\begin{itemize}
\item
In the case of intersecting supermodular $f$ we choose $r=2k_f-1$ 
and need $r(2k_f^2-3k_f+2)^\ell \leq n-r$, 
namely, $n \geq (2k_f-1)[(2k_f^2-3k_f+2)^\ell +1]$.
\item
If $f$ is obtained by zeroing an intersecting supermodular function $g$ on  co-void bisets we choose $r=k_g$ 
and need $r(2k_g^2-3k_g+2)^\ell \leq n-r$, 
namely, $n \geq k[(2k_g^2-3k_g+2)^\ell+1]$.
\item
When $f=f_{k{\sf\mbox{-}CS}}$, \cite{FNR} shows a choice of $R_1$ such that $|R_2| \leq k^3-k$.
We need $(k^3-k)(2k^2-3k+2)^{\ell-1} \leq n-k$,
namely $n \geq k[(k^2-1)(2k^2-3k+2)^{\ell-1}+1]$.
\end{itemize}

To get a polynomial time implementation we need to find in step~4 the set $R_{i+1}=R_i \cup U(f^I,k_f)$ in polynomial time.
We modify the algorithm by relaxing the step~4 condition $R_{i+1}=R_i \cup U(f^I,k_f)$ to 
$R_i \subseteq R_{i+1} \subseteq R_i \cup U(f^I,k_f)$ (so $R_1 \subs R_2 \subs \ldots$ will be a nested family),
but require that for each $J_i=I \cup I'$ the algorithm will compute a cover $F_i$ of $f^{J_i}$ of cost $c(F_i) \leq 2\tau(f^{J_i})$.
This can be done in the same way as in \cite{CVegh}, as follows.

The iterative rounding $2$-approximation algorithm of \cite{FJW} for covering a posi\-tively skew supermodular biset function, 
when applied on an arbitrary biset function $h$, 
either returns a $2$-approximate cover $J$ of $h$, or a {\bf failure certificate}: 
a pair $\A,\B$ of bisets with $h(\A) >0$ and $h(\B)>0$ for which both the supermodular and the co-supermodular inequality does not hold.
In our case this can happen only if $\A,\B$ are independent, by Lemma~\ref{l:co}. 

Now consider some iteration $i$ of the algorithm.
Since $f^{I \cup I'}$ is symmetric, then by interchanging the roles of $\A,\A^*,\B,\B^*$,
we can assume w.l.o.g. that our failure certificate $\A,\B$ satisfies $A \subs \p\B$.
We thus apply the following procedure.
Start with $R_{i+1}=R_i$.
Then iteratively, find $I'$ as in step~5 and apply the $2$-approximation algorithm of \cite{FJW} for covering $h=f^{I \cup I'}$;
if the algorithm returns an edge set $F$ of cost $c(F) \leq 2\tau(h)$, 
we keep the current $R_{i+1}$,  set $J_i \gets I \cup I'$ and $F_i \gets F$, and continue to the next iteration.
Else, we have a failure certificate pair $\A,\B$ of $h$-positive bisets with $A \subs \p\B$. 
Then $A \subs U(f^I,k_f)$ and $A \sem R_{i+1} \neq \empt$ (since $I'$ $f$-covers bisets whose inner part is contained in $R_{i+1}$),
and we can apply the same procedure with a larger candidate set $R_{i+1} \gets R_{i+1} \cup A$.

This concludes the proof of Theorem~\ref{t:bfuc}.

% \bibliographystyle{abbrv}
% \bibliography{k-CS}

\begin{thebibliography}{10}

\bibitem{ACL}
A.~Aazami, J.~Cheriyan, and B.~Laekhanukit.
\newblock A bad example for the iterative rounding method for mincost
  $k$-connected spanning subgraphs.
\newblock {\em Discrete Optimization}, 10(1):25--41, 2013.

\bibitem{ADNP}
V.~Auletta, Y.~Dinitz, Z.~Nutov, and D.~Parente.
\newblock A 2-approximation algorithm for finding an optimum
  $3$-vertex-connected spanning subgraph.
\newblock {\em J. of Algorithms}, 32(1):21--30, 1999.

\bibitem{CT00}
J.~Cheriyan and R.~Thurimella.
\newblock Approximating minimum-size $k$-connected spanning subgraphs via
  matching.
\newblock {\em SIAM J. on Computing}, 30(2):528--560, 2000.

\bibitem{CVegh}
J.~Cheriyan and L.~V\'{e}gh.
\newblock Approximating minimum-cost $k$-node connected subgraphs via
  independence-free graphs.
\newblock {\em SIAM J. on Computing}, 43(4):1342--1362, 2014.

\bibitem{CVV}
J.~Cheriyan, S.~Vempala, and A.~Vetta.
\newblock An approximation algorithm for the minimum-cost $k$-vertex connected
  subgraph.
\newblock {\em SIAM J. on Computing}, 32(4):1050--1055, 2003.

\bibitem{DN}
Y.~Dinitz and Z.~Nutov.
\newblock A $3$-approximation algorithm for finding optimum
  $4,5$-vertex-connected spanning subgraphs.
\newblock {\em J. of Algorithms}, 32(1):31--40, 1999.

\bibitem{EV}
A.~Ene and A.~Vakilian.
\newblock Improved approximation algorithms for degree-bounded network design
  problems with node connectivity requirements.
\newblock In {\em STOC}, pages 754--763, 2014.

\bibitem{FL}
J.~Fakcharoenphol and B.~Laekhanukit.
\newblock An {$O(\log^2 k)$}-approximation algorithm for the $k$-vertex
  connected spanning subgraph problem.
\newblock {\em SIAM J. on Computing}, 41(5):1095--1109, 2012.

\bibitem{FJW}
L.~Fleischer, K.~Jain, and D.~Williamson.
\newblock Iterative rounding 2-approximation algorithms for minimum-cost vertex
  connectivity problems.
\newblock {\em J. Comput. Syst. Sci}, 72(5):838--867, 2006.

\bibitem{F-R}
A.~Frank.
\newblock Rooted $k$-connections in digraphs.
\newblock {\em Discrete Applied Mathematics}, 157(6):1242--1254, 2009.

\bibitem{Frank-book}
A.~Frank.
\newblock {\em Connections in Combinatorial Optimization}.
\newblock Oxford University Press, 2011.

\bibitem{FJ}
A.~Frank and T.~Jord\'{a}n.
\newblock Minimal edge-coverings of pairs of sets.
\newblock {\em J. on Comb. Theory B}, 65:73--110, 1995.

\bibitem{FK-survey}
A.~Frank and T.~Kiraly.
\newblock A survey on covering supermodular functions.
\newblock In W.~Cook, L.~Lovász, and J.~Vygen, editors, {\em Research Trends in
  Combinatorial Optimization}, pages 87--126. Springer, Berlin, 2009.

\bibitem{FT}
A.~Frank and E.~Tardos.
\newblock An application of submodular flows.
\newblock {\em Linear Algebra and its Applications}, 114/115:329--348, 1989.

\bibitem{FJ82}
G.~Fredrickson and J.~J\'{a}j\'{a}.
\newblock On the relationship between the biconnectivity augmentation and
  traveling salesman problem.
\newblock {\em Theorethical Computer Science}, 19(2):189--201, 1982.

\bibitem{FNR}
T.~Fukunaga, Z.~Nutov, and R.~Ravi.
\newblock Iterative rounding approximation algorithms for degree-bounded
  node-connectivity network design.
\newblock {\em SIAM J. on Computing}, 44(5):1202--1229, 2015.

\bibitem{GGPS}
M.~Goemans, A.~Goldberg, S.~Plotkin, D.~Shmoys, E.~Tardos, and D.~Williamson.
\newblock Improved approximation algorithms for network design problems.
\newblock In {\em SODA}, pages 223--232, 1994.

\bibitem{JJ1}
B.~Jackson and T.~Jord\'{a}n.
\newblock A near optimal algorithm for vertex connectivity augmentation.
\newblock In {\em ISAAC}, pages 313--325, 2000.

\bibitem{JJ2}
B.~Jackson and T.~Jord\'{a}n.
\newblock Independence free graphs and vertex connectivity augmentation.
\newblock {\em J. of Comb. Theory B}, 94:31--77, 2005.

\bibitem{Jain}
K.~Jain.
\newblock A factor 2 approximation algorithm for the generalized {S}teiner
  network problem.
\newblock {\em Combinatorica}, 21(1):39--60, 2001.

\bibitem{KR}
S.~Khuller and B.~Raghavachari.
\newblock Improved approximation algorithms for uniform connectivity problems.
\newblock {\em Journal of Algorithms}, 21:434--450, 1996.

\bibitem{KN1}
G.~Kortsarz and Z.~Nutov.
\newblock Approximating node-connectivity problems via set covers.
\newblock {\em Algorithmica}, 37:75--92, 2003.

\bibitem{KN2}
G.~Kortsarz and Z.~Nutov.
\newblock Approximating $k$-node connected subgraphs via critical graphs.
\newblock {\em SIAM J. on Computing}, 35(1):247--257, 2005.

\bibitem{LaN}
Y.~Lando and Z.~Nutov.
\newblock Inapproximability of survivable networks.
\newblock {\em Theoretical Computer Science}, 410(21-23):2122--2125, 2009.

\bibitem{NRS}
V.~Nagarajan, R.~Ravi, and M.~Singh.
\newblock Simpler analysis of {LP} extreme points for traveling salesman and
  survivable network design problems.
\newblock {\em Oper. Res. Lett.}, 38:156--160, 2010.

\bibitem{N-R}
Z.~Nutov.
\newblock Approximating minimum-cost connectivity problems via uncrossable
  bifamilies.
\newblock {\em ACM Trans. Algorithms}, 9(1):1:1--1:16, 2012.

\bibitem{N-small}
Z.~Nutov.
\newblock Small $\ell$-edge-covers in $k$-connected graphs.
\newblock {\em Discrete Applied Math.}, 161(13-14):2101--2106, 2013.

\bibitem{N-C}
Z.~Nutov.
\newblock Approximating minimum-cost edge-covers of crossing biset-families.
\newblock {\em Combinatorica}, 34(1):95--114, 2014.

\bibitem{N-A}
Z.~Nutov.
\newblock Improved approximation algorithms for minimum cost node-connectivity
  augmentation problems.
\newblock {\em Theory Comput. Syst.}, 62(3):510--532, 2018.

\bibitem{N-kcs}
Z.~Nutov.
\newblock The $k$-connected subgraph problem.
\newblock In T.~Gonzalez, editor, {\em Approximation Algorithms and
  Metaheuristics}, chapter~12, pages 213--232. Chapman \& Hall, 2018.

\bibitem{N-sn}
Z.~Nutov.
\newblock Node-connectivity survivable network problems.
\newblock In T.~Gonzalez, editor, {\em Approximation Algorithms and
  Metaheuristics}, chapter~13. Chapman \& Hall, 2018.

\bibitem{RWe}
R.~Ravi and D.~P. Williamson.
\newblock Erratum:~an~approxi\-mation algorithm for minimum-cost
  vertex-connectivity problems.
\newblock {\em Algorithmica}, 34(1):98--107, 2002.

\bibitem{Vegh}
L.~V\'{e}gh.
\newblock Augmenting undirected node-connectivity by one.
\newblock {\em SIAM J. Discrete Math.}, 25(2):695--718, 2011.

\bibitem{Vetta-scs}
A.~Vetta.
\newblock Approximating the minimum strongly connected subgraph via a matching
  lower bound.
\newblock In {\em SODA}, pages 417--426, 2001.

\end{thebibliography}

\end{document}